\documentclass[aps,prl,twocolumn,showpacs,superscriptaddress,groupedaddress]{revtex4}
\usepackage[centertags]{amsmath}
\usepackage{amsmath}
\usepackage[dutch,british]{babel}
\usepackage{amsfonts}
\usepackage{amssymb}
\usepackage{amsthm}
\usepackage{newlfont}
\usepackage{color}
\usepackage{subfig}
\usepackage{bbm}
\usepackage{bm}

\usepackage{stmaryrd}
\usepackage{mathrsfs}
\usepackage{fancyhdr}
\usepackage{graphicx}
\usepackage{fancybox}
 \usepackage{psfrag}
 
 \usepackage{multirow}

\newcommand{\caB}{{\mathcal B}}

\newcommand{\caH}{{\mathcal H}}

\newcommand{\caN}{{\mathcal N}}
\newcommand{\caO}{{\mathcal O}}
\newcommand{\caP}{{\mathcal P}}

\newcommand{\caT}{{\mathcal T}}

\newcommand{\caV}{{\mathcal V}}

\newcommand{\bbC}{\mathbb{C}}
\newcommand{\bbN}{\mathbb{N}}

\newcommand{\beq}{\begin{equation}}
\newcommand{\eeq}{\end{equation} }

\newcommand{\e}{\mathrm{e}}
\newcommand{\iu}{\mathrm{i}}
\renewcommand{\d}{\mathrm{d}}

\newcommand{\norm}{ ||}
\newcommand{\bbZ}{\mathbb{Z}}

\theoremstyle{plain}
\newtheorem{theorem}{Theorem}
\newtheorem{lemma}{Lemma}

\newcommand{\ad}{\mathrm{ad}}
\newcommand{\adjoint}{\mathrm{ad}}
\newcommand{\diam}{\mathrm{diam}}

\newcommand{\bfN}{\mathbf{N}}
\newcommand{\bfJ}{\mathbf{J}}
\newcommand{\bfm}{\mathbf{m}}

\begin{document}

\title{Very slow heating for weakly driven quantum many-body systems}

\author{Wojciech De Roeck}
\affiliation{KU Leuven,
3001 Leuven, Belgium}
\email{wojciech.deroeck@kuleuven.be}

\author{Victor Verreet}
\affiliation{KU Leuven,
3001 Leuven, Belgium}
\email{victor.verreet@kuleuven.be}


\date{\today}

\vspace{0.5cm}
\begin{abstract}
It is well understood that many-body systems driven at high frequency heat up only exponentially slowly and exhibit a long prethermalization regime. We prove rigorously that a certain relevant class of systems heat very slowly under weak periodic driving at intermediate frequency as well.  This class of systems are those whose time-dependent, possibly translation-invariant, Hamiltonian is a weak perturbation of a sum of mutually commuting,  terms.   This condition covers several periodically kicked systems that have been considered in the literature recently, in particular kicked Ising models. In contrast to the high-frequency regime, the prethermalization dynamics of our systems is in general not related to any time-independent effective Floquet Hamiltonian. 
Our results also have non-trivial implications for closed (time-independent) systems. We use the example of an Ising model with transversal and longitudinal field to show how they imply confinement of excitations. More generally, they show how ``glassy"  kinetically constrained models emerge naturally from simple many-body Hamiltonians, thus connecting to the topic of  'translation-invariant localization'. 
\end{abstract}

\maketitle

\noindent \textbf{Introduction} The phenomenon of a long-lived quasi-stationary state, also known as a \emph{prethermal state},  has become an important paradigm in the theory of non-equilibrium many-body systems.  Such quasi-stationary states often exhibit interesting features that cannot be realized in equilibrium states \cite{martin2017topological,lindner2017universal,yao2017discrete,else2017prethermal}.  One crucial ingredient to have a long-lived prethermal regime is of course that the system resists thermalization for a long time.   
In this letter, we are interested in cases where the thermalization rate is beyond perturbation theory. In particular, we do not consider the oft-discussed case of prethermalization in weakly perturbed integrable systems, where the prethermalization regime lasts in general for a time $\propto 1/g^2$ with $g$ the perturbation strength \cite{berges2004prethermalization,bertini2015prethermalization,eckstein2009thermalization,mallayya2018prethermalization,reimann2019typicality}, but instead we look for thermalization times superpolynomial in $1/g$.

Several instances of such prethermal phases have been identified. 
The most obvious class consists of systems periodically driven at high frequency, where a Magnus-Floquet expansion yields an effective Hamiltonian \cite{abanin2017rigorous,lazarides_das_14,dalessio_rigol_14,mori2016rigorous,abanin2017effective,bukov2015universal}. 
It is natural to inquire whether a similar scenario might also be realized away from the high-frequency regime and in this letter we aim to provide a positive and general answer to this question.

\noindent \textbf{Weakly driven systems}
Given a Hamiltonian $H=H^0+g W(t)$ describing spins on a spatial lattice of arbitrary dimension, with periodic but otherwise generic driving $W(t)=W(t+T)$. When is the heating rate non-perturbatively small in $g$, i.e.\ taking $g$ small compared to other local energy scales, regardless of the initial state?  This question remains meaningful if we allow ${H^{0}}$ to depend periodically on $t$ as well, as long as its one-cylce propagator $U_0=\caT\e^{-\iu \int^T_0 dt {H^{0}}(t)}$ ($\caT\ldots$ denotes time-ordering of the epxonential), or one of its powers $U_0^p, p \in \bbN$,    has a meaningful (i.e.\ local) conservation law, i.e.\ there is a $Q$  such that $[Q,U_0^p]=0$ and $Q$ is a sum of local terms. Such a conservation law has observable consequences and the question of its persistence at $g\neq 0$ is well-posed. 
In this setting, we argue that heating is non-pertubatively slow whenever we can represent the $p$-cycle propagator as 
$U_0^p =\caT\e^{-\iu \int^{pT}_0 dt {\widetilde H^{0}}(t)}$ with ${\widetilde H^{0}}(t)$ a sum of local terms ${\widetilde H^{0}}_i(t)$ that mutually commute at all times $[{\widetilde H^{0}}_i(t),{\widetilde H^{0}}_j(t')]=0$, and a certain \emph{Diophantine condition} is satisfied.   

\noindent \textbf{Static systems}
Now we take a time-independent local many-body Hamiltonian $H=H^0+g W$ with $W$ generic and $g$ again small compared to local energy scales. Instead of 'heating', the appropriate question is now to what extent conservation laws of $H^0$ (in particular $H^0$ itself) are broken by $gW$.  We claim that a sufficient conditon for superpolynomial persistence of these conservation laws  is again that $H^0$ is a sum of mutually commuting terms. The larger the set of Bohr frequencies (energy differences)  defined by these local commuting terms, the smaller we need to take $g$ to see the effect, but the more numerous the number of quasi-conserved charges $Q$. This multitude of conserved quantities can lead to local frustration and emergent kinetic constraints, as we will demonstrate in an example.

\noindent \textbf{Ergodicity-breaking phases} The most well-known examples of the above claims are the cases where $H^0,\widetilde H^0$, respectively, are sums of commuting disordered terms (known as 'LIOMs' \cite{imbrie2016many,serbyn2013local,huse2014phenomenology,ros2015integrals}
). Then, those claims are weakened versions of the stability of \emph{many-body localization} (MBL) w.r.t.\ small perturbations, be it in static systems \cite{basko2006metal,gornyi2005interacting,nandkishore2014many}, periodically driven \cite{ponte2015many,lazarides2015fate,abanin2016theory} or time-crystals (case $p>1$)\cite{wilczek2012quantum,khemani2016phase,khemani2017defining}. 
Our claims are weaker because they only state that the thermalization time is very large instead of infinite and also because they posit the existance of a few ($\caO(1)$ regardless of volume)  quasi-conserved charges $Q$.
In this letter,  we concentrate on translation invariant $H^0,\widetilde H^0$ and we exhibit rigorously $\caO(1)$ number of charges. 
Depending on the initial state, it may actually happen that there emerge additional $\caO(L)$ quasi-conserved charges ($L$= number of degrees of freedom) via frustration and effective kinetic constraints. This is related to \emph{fractons} \cite{prem2017glassy,pai2019localization} and, more generally, to so-called translation invariant quasi-localization, see e.g.\ \cite{kagan1984localization,de2014asymptotic,schiulaz2015dynamics,
yao2014quasi,hickey2016signatures}.

\noindent \textbf{Intuitive picture}
Let us start from the static setting $H=H^0+gW$.
As announced, we assume that $H^0=\sum_i H^0_i$ with $[H^0_i,H^0_j]=0$, therefore every operator of the form $\sum_i f(H^0_i)$ commutes with $H^0$ as well. If now all $H^0_i$ have a finite spatial range $R$ and norm bound $\gamma$,  and the number of degrees of freedom per site (local Hilbert space dimension) is finite as well, then we can represent
$$
H^0=\bfJ\cdot\bfN=\sum_{\alpha=1}^q J_{\alpha}N^{({\alpha})}
$$
where $N^{({\alpha})}=\sum_i N^{({\alpha})}_i$ with $i$ labelling lattice sites and such that
\textbf{a})  $[N^{({\alpha})}_i,N^{({\alpha}')}_j]=0$, \textbf{b}) $N^{({\alpha})}_i$ has integer spectrum, \textbf{c}) $N^{({\alpha})}_i$ acts only within a distance $R$ of site $i$ and is uniformly bounded $\norm N_i^\alpha\norm\leq \gamma$.\\ 
If the conservation laws $N^{(\alpha)}$ are to be broken by local terms of strength $g$, we have to find \emph{local transitions} $ \mathbf{n}\to\mathbf{n}+ \Delta \mathbf{n}$ with $\Delta \mathbf{n} \in \bbZ^q$ that are on-resonance up to an energy offset of order $g$, i.e.\ 
\begin{equation}\label{eq: nonres condition}
|\mathbf{J}\cdot \Delta \mathbf{n}| \leq \caO(g) 
\end{equation}
Since $gW$ acts locally and the $N_i^\alpha$ are bounded, the components of $ \Delta \mathbf{n}$ have to be smaller than some fixed value. This means that, for generic $\bm{J}$ and small enough $g$, this constraint allows no nonzero solution $\mathbf{n}\neq 0$ and hence we are led to believe that all $N^{(\alpha)}$ are conserved. This is of course merely first-order reasoning but it turns out (see later) that higher orders do not change the picture, upon taking non-dissipative effects into account by a unitary frame rotation.  
If we add periodic time-dependence (frequency $2\pi/T$) to the problem, then similar reasoning applies with the role of $H^0$ now played by $\frac{1}{T}\int_0^{pT}dt H^0(t)$. To allow for the possibility of absorbing/emitting $\Delta n_0$ quanta from the drive, we modify the condition \eqref{eq: nonres condition} to 
$$
|\mathbf{J}\cdot \Delta \mathbf{n} + (2\pi/T)\Delta n_0 | \leq g 
$$
The same considerations as above apply, except that now we need a constraint on the $q+1$-tuple $(2\pi/T,\bm{J}) $.
Such constraints are well-known from KAM-phenomena \cite{poschel2001lecture} and Nekoroshev estimates, where, just as in our case, they express that resonances are absent. 
Since the frequency $2\pi/T$ enters the above formulas on the same footing as the couplings $J^\alpha$, the same formalism applies to quasi-periodically driven systems as well (long prethermalization observed in \cite{dumitrescu2018logarithmically}), see \cite{else2019long}. 

We stress that the above considerations involve \emph{spatial locality}, as opposed to locality in momentum space.    Indeed, consider the relevant case of free fermions:
${H^{0}}=\sum_k \omega(k) n_k$ with $n_k=c^\dagger_k c_k$ the occupation operator of momentum mode $k$, and $\omega(k)$ the dispersion relation. Hence ${H^{0}}$ is a sum of very simple and mutually commuting operators, but they are \textbf{not} local and they do not satisfy our requirements, 
Indeed, for such $H^0$ one expects kinetic theory based on non-linear Boltzmann equations to describe the heating process, with rate $\propto g^2$, see \cite{spohn2007kinetic,lukkarinen2011weakly}.

\noindent \textbf{Result for static systems} 
Our systems live on a large but finite graph $\Lambda$, with a finite Hilbert space $\bbC^d$ attached to each site $i\in \Lambda$.  The total  Hilbert space $\caH$ is hence $(\bbC^d)^{\otimes_\Lambda}$ and we say that an operator $O=O_S$ is supported in a set $S$ if it is of the form $ O_S \otimes \mathbbm{1}_{S^c}$. We consider a Hamiltonian of the form
$$
H=H^0+gW,
$$
with 
\begin{equation}\label{def: hzero}
H^0= \bfJ\cdot\bfN= \sum_{\alpha=1}^q J_\alpha N^{(\alpha)}
\end{equation}
where the  operators $N^{(\alpha)}$ satisfy the conditions \textbf{a,b,c}) in the previous section, and, as a concrete generalization of condition \eqref{eq: nonres condition} we assume the following \emph{Diophantine} condition expressing that the $\bfJ$ are sufficiently incommensurable at all orders in perturbation theory: there exist positive numbers $\tau, x>0$, such that
\begin{equation}\label{eq: diophantine static}
|\bfm\cdot{\bfJ}| \geq \frac{x |\bfJ|}{|\bfm|^{\tau}},\qquad \forall \bfm \in \bbZ^q
\end{equation}
Finally, we need the Hamiltonian $gW$ to have local terms of strength  $\caO(g)$, decaying fast in the size of their support. To that end, we define (see SM) the \emph{local norms} $\norm \cdot \norm_\kappa $ with spatial decay parameter $\kappa>0$   and we choose a $\kappa_0>0$ such that $\norm gW\norm_{\kappa_0} \leq g$. 

\begin{theorem}
There is a unitary base change  $\hat Y$  such that, in the rotated frame, the Hamiltonian   takes the form
$$  \hat{Y}H\hat{Y}^\dagger= {H^{0}}+\hat{D}+\hat{V} $$
where, with constants $c,C$ depending only on the  parameters $\kappa_0,\gamma, q,p$ and in particular not on the volume $|\Lambda|$, 
\begin{enumerate}
\item  The driving $\hat V$ is very weak: 
$$\norm \hat{V}\norm_{\kappa}  \leq  g \exp{(-c[1/\epsilon]^{1/p})},\qquad  \epsilon \equiv \frac{g}{|\bfJ|}(1+\tfrac1x)   $$
for any $p>q+\tau$ and $\kappa=c\kappa_0|\log\epsilon|^{-1}$. 
\item  $\norm \hat{D}\norm_{\kappa}  \leq Cg $ and $[N^{(\alpha)},\hat{D}]=0$ for any $\alpha =1,\ldots,q$.  
\end{enumerate}
\end{theorem} 
We see hence that $\hat D$ inherits the conservation laws of $H^0$ and $\hat V$ should be considered as a remaining weak driving, enforcing thermalization at quasi-exponential times.  
Importantly, the unitary conjugation $\hat Y \cdot \hat Y^\dagger$ is close to identity: for local operators $O$ with support size $\caO(1)$, we have $\norm \hat Y O \hat Y^\dagger -O \norm \leq C\epsilon\norm O\norm$ (see SM for a more precise statement and construction). 
One obvious consequence is that, in the original frame, the dressed operators $\hat N^{(\alpha)}\equiv \hat Y N^{(\alpha)} \hat Y^\dagger$, which are sums of quasi-ocal terms, are quasi-conserved:
$$
\tfrac{1}{|\Lambda|}\norm \e^{\iu t H}\hat N^{(\alpha)}\e^{-\iu t H} -\hat N^{(\alpha)}\norm \leq  g \exp{(-c[1/\epsilon]^{1/p})} 
$$
up to stretched-exponential long time $\exp{(c[1/\epsilon]^{1/p})} $.   Hence in particular thermalization is obstructed until that time.\\ 
\noindent \textbf{Comments}
If we imagine $\bfJ$ to be chosen uniformly on the unit hypersphere $S^{q}$, then the probability of the Diophantine condition \eqref{eq: diophantine static} being violated, decays as $Cx$ when $x\to0$, provided that $\tau>q-1$. In particular, \eqref{eq: diophantine static} is satisfied with probability $1$ for some $x$. 
It follows that the power $p$ in the stretched exponential can be chosen $p=2q-1+\delta$, with $\delta$ arbitrarily small.
In particular, for the case $q=1$, this yields almost an exponential in $1/\epsilon$, as was already proven in \cite{abanin2017rigorous,else2017prethermal}.

If the condition \eqref{eq: diophantine static} (for a given $x$) fails for some $\bfm$ with $|\bfm| \sim R\gamma n$, then this simply means that our perturbative reduction of the driving $\hat V$ can only be carried out to a power $n$ instead of a power diverging as $\epsilon\to 0$, giving $\norm V\norm_{\kappa} \leq g \epsilon^{n}$. Therefore, it is really only low-order resonances that can hamper the slow thermalization.  
However, such low-order resonances can almost always be lifted by redefining $H^0$, as we also illustrate in the example.\\
\noindent \textbf{Example: static Ising}
Let
$$
H= \sum_i J\sigma^x_i \sigma^x_{i+1}+ h_x \sigma^x_i+h_z\sigma^z_i. $$
and we consider $h_z$ as the weak-driving parameter $g$, i.e.\ we set  $H^0 \equiv JN^{(1)}+h_x N^{(2)} $ with $N^{(1)}\sum_i \sigma^x_i \sigma^x_{i+1}$  the number of domain walls (up to a constant) and  $N^{(2)}=\sum_i \sigma^x_i$ the magnetization. 

If the pair  $(J,h_x)$ is sufficiently incommensurable, i.e.\ condition \eqref{eq: diophantine static} holds, then our theory yields that, in the rotated frame, the Hamiltonian is approximately $\hat H= JN^{(1)}+h_x N^{(2)}+\hat D $ ( we neglected the very weak driving $\hat V$), with both $N^{(1,2)}$ commuting with $\hat D$ and hence conserved locally.
For example, let us consider a configuration of the form
$$\ldots \downarrow\downarrow\downarrow\downarrow  \underbrace{\uparrow\uparrow \cdot\uparrow\uparrow}_{\text{length}\,\, \ell_0}\downarrow\downarrow\downarrow\downarrow  \ldots $$
The only transition that preserves both magnetization and the number of domains consists of a shift of the entire block $\uparrow\uparrow \cdot\uparrow\uparrow$. However, since the operator $\hat D$ is a sum of local terms, the largest term that can cause such a shift, is of size $g\epsilon^{c\ell_0}$ with $c$ a constant of order $1$. 
In our formalism, this follows by the local bound $\norm \hat D\norm_{\kappa} \leq Cg$, expressing exponential decay of local terms in $\hat D$. Hence, the domains acquire a very large mass and and are nearly static. 
An even more drastic example is
$$\ldots \downarrow\downarrow\downarrow\downarrow\downarrow\downarrow  \uparrow\uparrow \uparrow\uparrow \uparrow\uparrow \ldots $$
with the configuration extending infinitely far in both directions. Here, \emph{no single transition} is allowed by the conservation laws and there should be therefore no dissipative dynamics up to (quasi-)exponential time.   
These phenomena were numerically studied and explained in \cite{mazza2019suppression}, our theorem adds a controlled proof and a novel point of view.  
If $(J,h_x)$ are commensurable, say $J=2h_x$, then of course our theorem does not apply in the above way. However, in that case, one can choose $H^0=h_x N^{(1')}$ with $N^{(1')}\equiv 2N^{(1)}+N^{(2)} $, giving again a meaningful constraint.
Let us finally return to the case of incommensurable  $(J,h_x)$. More generally, if we consider a system with a small density of domain walls, the above considerations show that the local dynamics is highly constrained and appears many-body-localized for a long time. It can hence be a considered as an emergent kinetically constrained model \cite{van2015dynamics}.  However, as was argued in \cite{de2014scenario}, the long time scale that emerges here is in general not beyond perturbation theory in the parameter $g$.

\noindent \textbf{Result for periodic driving}
We are inspired by the setup proposed in the section "weakly driven systems", with $U_0^p$ generated by $\widetilde H^0(t)$. However, for the sake of notational simplicity, we write now $H^0$ instead of $\widetilde H^0$ and we redefine $T \mapsto pT$. This means that we potentially describe the original system of interest only at stroboscopic times $t \in pT\bbN$, but this suffices for the sake of obstructions to thermalization. \\
Concretely, we consider a $T$-periodic Hamiltonian of the form 
$$
H(t)=H^0(t)+gW(t), \qquad  H(t)=H(t+T)
$$
with 
\begin{equation}\label{def: hzero}
H^0(t)= \bfJ(t)\cdot\bfN= \sum_{\alpha=1}^q J_\alpha(t)N^{(\alpha)}
\end{equation}
Here, the  operators $N^{(\alpha)}$ are exactly as in the static setup with the sole difference that we allow $T$-periodic coupling $\bfJ(t)=\bfJ(t+T)$.  Furthermore, we assume that the time-dependent (local terms of) $W(t)$ and the couplings $\bfJ(t)$ are piecewise-continuous. 
Just as in the static case, we assume the local bound  $\norm g W\norm_{\kappa_0} \leq g$, except that now this bound includes a supremum over $t$, see SM. Finally, the \emph{Diophantine} condition is slighly modified: there are $0<\tau,x<\infty$ such that, with
$\bar{\bm{J}}=\tfrac1T\int_0^T dt \bm{J}(t)$
\begin{equation}
\label{eq: diophantine periodic}
\inf_{n\in\bbZ}|T\bfm\cdot\bar{\bfJ}- 2\pi n| \geq \frac{x}{|\bfm|^{\tau}},\qquad \forall \bfm \in \bbZ^q
\end{equation}
We let $U(t)$ be the solution of the Schrodinger equation
$$
\iu \partial_t U(t)=H(t)U(t),\qquad  U(0)=1
$$
\begin{theorem}\label{thm: periodic}
There is a $T$-periodic unitary base-change  $\hat Y(t)$ such that  $
\hat U(t)\equiv \hat{Y}(t)U(t)$ solves the Schrodinger equation $$\iu \partial_t \hat{U}(t)= ({H^{0}}(t)+\hat{D}(t)+\hat{V}(t)) \hat{U}(t),$$
and  (with $c,C$ depending only on $\gamma,\kappa_0, q$ and $p$)
\begin{enumerate}
\item  The driving $\hat V(t)$ is very weak
 $$ \norm \hat{V}(t)\norm_{\kappa}  \leq g \exp{(-c[1/\epsilon]^{1/p})},\qquad \epsilon \equiv gT(1+\frac1x)$$
 with $p>q+\tau+1$ and  $\kappa=c\kappa_0|\log\epsilon|^{-1}$
\item  $ \norm \hat{D}(t)\norm_{\kappa}  \leq Cg $ and 
 $[N^{(\alpha)},\hat{D}]=0$ for any $\alpha =1,\ldots,q$. 
\end{enumerate}
\end{theorem}
Just as in the static case, this theorem shows that the dressed operators
$
\hat N^{(\alpha)}(t)= \hat{Y}(t)N^{(\alpha)} \hat{Y}^*(t)
$
are quasi-conserved
$
\tfrac{1}{|\Lambda|}\left(\norm \hat N^{(\alpha)}(t)- U(t)\hat N^{(\alpha)}(0)U^\dagger(t)\norm\right) \leq g \exp{(-c[1/\epsilon]^{1/p})} 
$ for $t \leq (1/g) \exp{(c[1/\epsilon]^{1/p})}$ 
and they obstruct heating.\\
Our scheme does not yield any \emph{effective Hamiltonian} $H_F$, satisfying $U\approx \e^{-\iu T H_F}$, with $\approx$ indicating validity up to long times.  Instead, we obtain 
$U= \caT \e^{-\iu \int_0^T dt (H^0(t)+ \hat{D}(t))} $. This would yield a $H_F=\tfrac{1}{T} \int_0^T dt (H^0(t)+ \hat{D}(t))$ \emph{provided that the local terms of $H^0(t)+ \hat{D}(t)$ commute between different times $t$}, but this is not true in general\footnote{The way to realize this is to convince oneself that the local terms in $\hat D(t)$ can be completely arbitrary in general, up to the commutation with $N^{(\alpha)}$} and so existence of some $H_F$ seems unlikely.

\begin{figure}
\begin{center}
\includegraphics[width=0.475\textwidth]{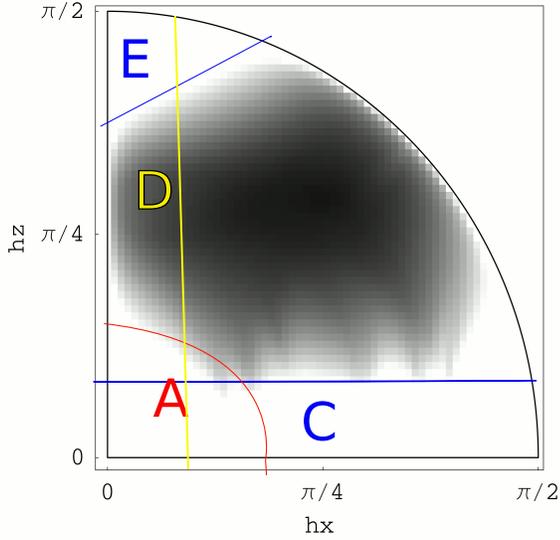}
\caption{(Taken From  \cite{prosen2007chaos}, Figure 13).  Kicked Ising chain at $J=1$. White regions are slowly thermalising: an appropriately defined thermalization rate is found to be smaller than $10^{-6}$ per cycle. The colored lines and letters demarcate our perturbative regimes}
\label{fig: kickedising}
\end{center}
\end{figure}

\noindent \textbf{Example: kicked Ising chain} We consider the one-cycle unitary given by  
$$U \equiv e^{-\iu \sum_i J\sigma^x_i \sigma^x_{i+1}} e^{-\iu \sum_i \mathbf{h}\cdot\bm{\sigma}_i }.
$$
where $\bm{\sigma}_i=(\sigma^{x}_i,\sigma^{y}_i,\sigma^{z}_i)$ are the Pauli-matrices acting on site $i=1,\ldots,L$.
This model was studied numerically in detail in \cite{prosen2007chaos}. Figure \ref{fig: kickedising} (copied from Figure 13 in \cite{prosen2007chaos}) shows the regimes where one observes very slow thermalization numerically.  We view this model as originating from a time-dependent $H(t)=H_11_{[0,1]}(t)+H_21_{[1,2]}(t)$ with period $T=2$. 
Our Theorem \ref{thm: periodic} applies to this model in several cases (not mutually exclusive), distinghuished by which Hamiltonian plays the role of $H^0(t)$.
\\
 
\noindent {\textbf{Case A ($|\textbf{h}|\ll 1$).}} Here, we let ${H^{0}}(t)=J(t)\sum_i\sigma^x_i \sigma^x_{i+1}$ with $J(t)=J1_{ [0,1]}(t)$.  Theorem \ref{thm: periodic} posits that heating is non-perturbatively slow in $|\mathbf{h}| \ll 1 $.\\
\noindent \textbf{Case B ($|J|\ll 1$).} Here, we take  ${H^{0}}(t)=\sum_i\mathbf{h}(t)\cdot\bm{\sigma}_i$ with $\mathbf{h}(t)= \mathbf{h}1_{ [1,2]}$.   Theorem \ref{thm: periodic} says that heating is no-perturbatively slow in $|J|\ll 1$. \\
 \noindent \textbf{Case C  ($|h_z|\ll 1$).}
Here 
$ {H^{0}}(t)=   \sum_i (J1_{[0,1]}(t)\sigma^x_i \sigma^x_{i+1} +  h_x1_{[1,2]}(t)\sigma_i^x)$ and we take  $|h_z| \ll 1$, again Theorem \ref{thm: periodic} shows slow heating in that regime. \\
\noindent \textbf{Case D ($|h_x|\ll 1$).}
Here $U_0$ can be mapped to a free fermion expression with non-trivial dispersion relation by the Jordan-Wigner transformation. However, as we argued in section 'intuitive picture', this does not meet our criteria.  The figure  shows indeed that there is \textbf{no} very slow heating for $|h_x|\ll 1$ (provided other couplings are not in a perturbative regime).  \\
\noindent \textbf{Case E ($|\bm{h}-(0,0,\tfrac{\pi}{2})|\ll 1$).}
Here we need to take the second power of the propagator. Indeed, now $U_0= e^{-\iu \sum_i J\sigma^x_i \sigma^x_{i+1}}e^{-\iu \tfrac\pi2 \sum_i \sigma^z_i}= e^{-\iu \sum_i J\sigma^x_i \sigma^x_{i+1}} F $ with the flip $F=\prod_i (-\iu\sigma_i^z)$. Therefore, $U_0^2= \e^{-2\iu J \sum_i \sigma^x_i \sigma^x_{i+1} }$, i.e.\ the setup applies simply setting $H^0(t)=2J \sum_i \sigma^x_i \sigma^x_{i+1}  $.

The claims in cases \textbf{A,B} have also been confirmed recently in \cite{vajna2018replica} by numerics and a replica expansion for an effective Hamiltonian $H_{F}$ satisfying $U=\e^{-\iu T H_{F}}$. As explained above, the existence of $H_F$ is however not suggested by our treatment. More precisely,\cite{vajna2018replica} predict slow thermalization with the stretched exponential with $p=2$.
In our theorem, taking $q=1$ and $\tau>1$ so that the Diophantine condition \eqref{eq: diophantine periodic} is satisfied almost surely for some $x$  (see SM), we prove these claims with $p=3+\delta$, for arbitrarily small $\delta$.

\noindent \textbf{Conclusion}
 We have identified a class of conditions under which non-perturbative rigorous lower bounds on the thermalization time can be proven, both in static and periodically driven systems. In the driven case, our theorems allow to understand the phase diagram of the kicked Ising model. In the static case, we provide a rigorous underpinning of the observed very slow dynamics of domain walls in Ising models.

\noindent \textbf{Note} The manuscript \cite{else2019long} appeared while we were finishing our paper. It is based on the same ideas as the present letter, has very analogous results, and goes beyond our analysis in many respects.  Since it however focuses on the case of quasiperiodic driving, its analogue of our Theorem \ref{thm: periodic} has technical restrictions on the smoothness of the driving protocol, in particular ruling out our main example ``kicked Ising models".
 
\begin{acknowledgments} \textbf{Acknowledgements}
 W.D.R. acknowledges the support of the Flemish Research Fund FWO undergrants G098919N and G076216N, and the support of KULeuven University under internal grantC14/16/062. 
\end{acknowledgments}

\bibliographystyle{plain}
\bibliography{loclibrary}
\newpage
\widetext
\begin{center}
\textbf{\Large Supplemental Material}
\end{center}
\setcounter{equation}{0}
\setcounter{figure}{0}
\setcounter{table}{0}
\setcounter{page}{1}
\makeatletter
\renewcommand{\theequation}{A\arabic{equation}}
\renewcommand{\thefigure}{A\arabic{figure}}

The aim of this SM is to provide a full proof of our results. For reasons of clarity, we repeat the full setup here. Some parameters are defined in a slightly different way, but we hope that the path to the statements in the main text is yet sufficiently direct. 
Our proofs are based on rigorous implementations of Schrieffer-Wolff transformations. These have a ppeared a lot in mathematical physics, and we have been in particular influenced by \cite{datta1996low,imbrie2016many}.  Cluster expansions are then used to resum the resulting expansions, and we use the neat formalism of \cite{ueltschi2004cluster}. From a more direct point of view, the proof is based on techniques used in \cite{abanin2017rigorous,else2017prethermal} that ultimately are descendants of KAM theory and Nekoroshev estimates \cite{poschel1993nekhoroshev}. The application of such ideas in systems with a few degrees of freedom is standard and we do not review it. There are also a few rigorous works where such ideas are used to describe many-body systems at spectral edges (i.e.\ low energy, low density) where the situation effectively reduces to few-body theory, see e.g.\ \cite{benettin1988nekhoroshev,mackay1994proof,cuneo2017energy}. 
Rigorous results where such techniques are applied to constrain the dynamics in a genuine many-body setup have come into focus in the last years, see e.g.\ \cite{huveneers2013drastic,de2014asymptotic,de2015asymptotic,else2017prethermal,giorgilli2015extensive,imbrie2016many}

\section{Technical preliminaries and setup}

\subsubsection{Spaces}

We consider a large but finite graph $\Lambda$, equipped with the graph distance. The vertices $i$ of this graph are our 'sites'.  There is a finite Hilbert space $\bbC^d$ attached to each site $i\in \Lambda$, and we take $d$ to be fixed. 
 The total  Hilbert space $\caH$ is hence $(\bbC^d)^{\otimes_\Lambda}$ and we say that an operator $O=O_S$ in $\caB\equiv\caB(\caH)$ (space of bounded operators on $\caH$) is supported in a set $S$ if it is of the form $ O_S \otimes \mathbbm{1}_{S^c}$, with a slight abuse of notation.   This is the setting for quantum spin systems. One can also consider lattice fermions, if one makes some modifications in the definition, see .
For operators $O \in \caB$, we use the standard operator norm $\norm O\norm=\sup_{\psi \in\caH,\psi \neq 0} \tfrac{\norm O\psi\norm}{\norm\psi\norm}  $  where, on the right hand side, $\norm\cdot\norm$ is the Hilbert space norm on $\caH$. 
  For an operator $A$, we will freely use the notation $\adjoint_A$ to denote the superoperator acting on $\caB$ as 
$$
\adjoint_A(B)=[A,B]
$$

 \subsubsection{The 'number' operators $N^{(\alpha)}, \alpha=1,\ldots, q$}

 These operators play a central role in our analysis.  They are given as sums of local terms $N=\sum_{S \subset \Lambda}N^\alpha_{S}$ satisfying the following conditions
\begin{enumerate}
\item All local terms mutually commute: $[N^\alpha_S, N^{\alpha'}_{S'}]$. 
\item All of the $N^\alpha_S$ have integer spectrum.
\item  There is a fixed range $R$ such that $N^\alpha_S=0$ whenever $\diam(S)>R$.
\item  There is a local bound $\sup_{i\in \Lambda}\sum_{S \ni i}\norm N^\alpha_S\norm  \leq \gamma$. 
\end{enumerate}
 As explained in the main text, item 1,2) actually follow from the assumption of having a single operator $H^0$ with commuting local terms.

With these definitions in hand, we need to refine the notion of support of operators, following \cite{else2017prethermal}.
 We say that $O \in \caB$ is 'strongly supported' in $S$ if
 $O$ is supported in $S$ and, for any $S' \not\subset S$ we have $[O,N^\alpha_{S'}]=0$.  
Here are the important consequences
\begin{enumerate}
\item For any function $f$, if $O$ is strongly supported in $S$, then 
$ f( \ad_N) O $ is strongly supported in $S$. 
\item If $A,B$ are strongly supported in $S_A,S_B$, then  $[A,B]$ is strongly supported in $S_A\cup S_B$. 
\end{enumerate}
We wirte $\caB_S \subset \caB$ for the algebra of operators strongly supported in $S$.

 \subsubsection{The projectors $\caP_{\bfm}$}
First, we introduce the superoperator $\caN^{(\alpha)}=\adjoint_{N^{(\alpha)}}$ acting on $\caB$, and, for any $\mathbf{m}\in \bbZ^q$
$$
\caP_\mathbf{m} (O)=  \prod_{\alpha=1}^{q}\frac{1}{2\pi}\int_0^{2\pi} \d s  \e^{\iu s (\caN^{(\alpha)}-m_\alpha)} (O).
$$
One checks that 
\begin{enumerate}
\item  $\caP_\mathbf{m}=\caP^2_\mathbf{m}$, i.e.\ $\caP_\mathbf{m}$ is a projector.
\item  $\caP_\mathbf{m}$ is a contraction in the operator norm: $||\caP_m(O)|| \leq || O||$. 
\item  If $O$ is strongly supported in $S$, then $\caP_\mathbf{m}(O)$ is strongly supported in $S$. 
\end{enumerate}

 \subsubsection{Norms on Hamiltonians}

To handle operators $G$ that are sums of local terms, like many-body Hamiltonians, we introduce \emph{local norms} $\norm\cdot \norm_\kappa$ that are defined pertaining to a representation of $G$ as a sum of local terms
$$
G=\sum_{S \in \caP_\Lambda} G_S,   \qquad  G_S \in \caB_S
$$
where $\caP_\Lambda$ stands for the set of connected (by adjacency) subsets of $\Lambda$.  Since a given operators $G$ can always be represented in different ways as a sum over local terms (e.g. $\sigma^z_i+\sigma^z_{i+1}$ can be viewed as one term with $S=\{i,i+1\}$ or as two terms with $S=\{i\}, \{i+1\}$), the above definition does not immediately yield a well-defined norm. One could try to take the supremum over all representations but that is cumbersome in practice.    Therefore, the standard solution to this problem in mathematical statistical mechanics is to define the functions $S\mapsto G_S $ as the central objects ("\emph{Potentials}" or "\emph{Interactions}", see e.g.\ \cite{simon2014statistical}) and to have norms on them, and we will do this here. However, to keep the notation light, we denote these objects by the same symbol $G$. 
A further complication is that we consider time-dependent Hamiltonians $G(t)$and hence potentials, but here one can simply take the supremum over $t$ and mostly drop $t$ from the notation.  The norm is then, for any decay rate $\kappa$, 
$$
\norm G\norm_\kappa \equiv \max_{i \in \Lambda}\sum_{S \in \caP_\Lambda, S\ni i}\e^{\kappa |S|}\sup_t\norm G_S(t) \norm
$$

\section{Statement of results}

\subsection{Time-dependent results}

We assume that our time-dependent Hamiltonian is of the form 
$$
H(t)=H^0(t)+gW(t),
$$
with 
\begin{equation}\label{def: hzero}
H^0(t)= \bfJ(t)\cdot\bfN= \sum_{\alpha=1}^q J_\alpha(t)N^{(\alpha)}
\end{equation}
where 
\begin{enumerate}
\item The operators $N^{(\alpha)}$ satisfy the conditions above, with parameters $R$ (range) and $\gamma$ (local bound).  
\item The local constituents $W_S(t)$ and the parameters $J(t)$ are periodic in $t$ and piecewise-continuous. 
\item  $\norm W\norm_{\kappa_0} \leq 1$. 
\item There are $\tau,x>0$ such that, with
$\bar{\bm{J}}\equiv \tfrac1T \int_0^T dt \bm{J}(t)$
\begin{equation} \label{eq: diophantine sm}
\inf_{n\in\bbZ}|T\bfm\cdot\bar{\bfJ}- 2\pi n| \geq \frac{x}{|\bfm|^{\tau}},\qquad \forall \bfm \in \bbZ^q
\end{equation}
\end{enumerate}

If  $\bfJ/|\bfJ|$ is sampled uniformly on the unit hypersphere, then such a constant $x$ can be found with probability $1$ provided that $q<\tau$, see the standard reasoning in the Appendix at the end 

The genuine dimensionless small parameter in our theory is then
$$
\epsilon \equiv gT(1+\tfrac1x)
$$

We let $U(t)$ be the solution of the Schrodinger equation
$$
\iu \partial_t U(t)=H(t)U(t),\qquad  U(0)=1
$$
\begin{theorem}
There is a $T$-periodic unitary  $\hat Y$ and $T$-periodic Hamiltonians $\hat{D}(t),\hat{V}(t)$ such that  $
\hat U(t)\equiv \hat{Y}(t)U(t)$ solves the Schrodinger equation $$\iu \partial_t \hat{U}= ({H^{0}}(t)+\hat{D}(t)+\hat{V}(t)) \hat{U}(t),$$
where $\hat{D}$ conserves $\bfN$ and $\hat{V}$ is very weak: 
\begin{enumerate}
\item 
 $[N^{(\alpha)},\hat{D}]=0 \quad \text{for any} \quad \alpha =1,\ldots,q. $ 
\item  $\norm\hat D \norm_\kappa \leq Cg$ and $
 \norm \hat{V}\norm_{\kappa}  \leq  g \e^{-C(1/\epsilon)^{1/p}}, $
 with $p$ any number larger than $q+\tau+1$, $\kappa=C\kappa_0|\log\epsilon|^{-1}$ and constants $C,c$ depending only on the parameters $\kappa_0,\gamma,q$ and on $p$. (These constants can have different values from line to line.)
 \item  The unitary transformation $\hat Y$ is quasi-local and close to identity in the sense that $\norm \hat Y G\hat Y^\dagger-G\norm_\kappa \leq  C\epsilon \norm G\norm_\kappa $ for any Hamiltonian $G$. 
\end{enumerate}
\end{theorem}

This theorem shows in particular that the dressed operators
$
\hat N^{(\alpha)}(t)= \hat{Y}(t)N^{(\alpha)} \hat{Y}^*(t)
$
are conserved up to very small errors, i.e.\ 
$$
\tfrac{1}{|\Lambda|}\left(\norm \hat N^{(\alpha)}(t)- U(t)\hat N^{(\alpha)}(0)U^\dagger(t)\norm\right) \leq C  g \e^{-c(1/\epsilon)^{1/p}},\qquad \text{for times $t \leq (c/g) \e^{c(1/\epsilon)^{1/p}} $}
$$
This follows by a simple Duhamel estimate, see \cite{abanin2017rigorous}.

\subsection{Static results}

We assume that our time-independent Hamiltonian is of the form 
$$
H=H^0+gW,
$$
with 
\begin{equation}\label{def: hzero static}
H^0(t)= \bfJ\cdot\bfN= \sum_{\alpha=1}^q J_\alpha(t)N^{(\alpha)}
\end{equation}
where 
\begin{enumerate}
\item The operators $N^{(\alpha)}$ satisfy the conditions above, with parameters $R$ (range) and $\gamma$ (local bound).  
\item  $\norm W\norm_{\kappa_0} \leq 1$. 
\item There are $tau,x>0y$ such that
$$
 \frac{1}{|\bfJ|}|\bfm\cdot{\bfJ}| \geq \frac{x}{|\bfm|^{\tau}},\qquad \forall \bfm \in \bbZ^q
$$
\end{enumerate}
If we sample $\bfJ/|\bfJ|$ uniformly on the hypersphere, then  a constant $x$ can be found with probability $1$ provided that $q-1<\tau$. This is of course essentially the same consideration as that following condition \eqref{eq: diophantine sm} (it would have been exactly the same, save for the dimension, if in\eqref{eq: diophantine sm} we were to treat $2\pi/T$ on the same footing as the $J_\alpha$).
The genuine dimensionless small parameter in our theory is then
$$
\epsilon \equiv \frac{g}{|\bfJ|}(1+\tfrac1x)
$$

\begin{theorem}
There is a  unitary  $\hat Y$ and Hamiltonians $\hat{D},\hat{V}$ such that  $$
\hat Y H \hat Y^\dagger= {H^{0}}+\hat{D}+\hat{V}$$
where again $\hat D$ conserves $\bfN$ and $\hat{V}$ is very weak:  
\begin{enumerate}
\item  $[N^{(\alpha)},\hat{D}]=0$ for any $\alpha =1,\ldots,q$. 
\item 
  $\norm \hat{D}\norm_{\kappa}  \leq C g$   and     $\norm \hat{V}\norm_{\kappa}  \leq  g \e^{-c(1/\epsilon)^{1/p}}$  with $p$ any number larger than $q+\tau+1$, $\kappa=C\kappa_0|\log\epsilon|^{-1}$ and $C;c$ depending only on the parameters $\kappa_0,\gamma,q$ and on $p$.  
 \item The unitary transformation $\hat Y$ is quasi-local and close to identity in the sense that $\norm \hat Y G\hat Y^\dagger-G\norm_\kappa \leq  C\epsilon \norm G\norm_\kappa $ for any Hamiltonian $G$. 
\end{enumerate}
\end{theorem}

This theorem shows in particular that the dressed operators
$
\hat N^{(\alpha)}(t)= \hat{Y}(t)N^{(\alpha)} \hat{Y}^\dagger(t)
$
are conserved up to very small errors, i.e.\ 
$$
\tfrac{1}{|\Lambda|}\left(\norm \hat N^{(\alpha)}(t)- U(t)\hat N^{(\alpha)}(0)U^\dagger(t)\norm\right) \leq C  g \e^{-c(1/\epsilon)^{1/p}},\qquad \text{for times $t \leq c/g \e^{c(1/\epsilon)^{1/p}} $}
$$
This follows by a simple Duhamel estimate, see \cite{abanin2017rigorous}.

\section{Proofs: algebra}

We primarily treat the time-periodic case, and afterwards we do the static case, focusing on the few differences.

\subsection{Renormalization of Hamiltonians}
We start from 
$$
H={H^{0}} +gW= {H^{0}}+V_1+ D_1
$$
where we put  $D_1=\caP_0(gW), V_1=gW-D_1$
Since we anticipate an interative procedure, we put immediately
$$
H_n={H^{0}}+V_n+ D_n,
$$
where $\caP_0(V_n)=0$ and $\caP_0(D_n)=D_n$. 
We perform a unitary transformation $e^{A_n}$ with $A_n$ anti-Hermitian, 
and we write
\beq  \label{eq: transformed hamiltonian}
e^{A_n} \left(\iu \partial_t+  {H^{0}}+D_n+V_n  \right) e^{-A_n} =   \iu \partial_t+  {H^{0}}+ D_n+ (\iu\partial_t A_n+[A_n,{H^{0}}] +V_n) + W_{n+1}
\eeq
where $W_n$ has been simply defined so as to the make the last equation correct, and we will later derive a precise expression. 

We now determine $A_n$ by the equation (written here immediately for arbitrary $n$)
\begin{equation}\label{eq: determine a}
\iu \partial_t A_n+  [ H^0,A_n]+V_n=0.
\end{equation}
choosing the solution that makes it periodic, namely 
\beq \label{eq: a in terms of v}
A_n(t)= \caV_{0,t} (1-\caV_{0,T})^{-1} \int_0^T \d s   \caV_{s,T} V_n(s)  + \int_0^t \d s   \caV_{s,t} V_n(s)
\eeq
where $\caV_{t_0,t}$ is the superoperator that is the solution of the equation
$$
\partial_{t}\caV_{t_0,t}  =\iu  \adjoint_{H^0}\caV_{t_0,t}
$$
The inverse of $1-\caV_{0,T}$ is well-defined on operators $O$ such that $\caP_0(O)=0$, by the Diophantine condition.  This is the reason we consider $V_n$ instead of $W_n$.
To proceed, we introduce some more notation.
We define the transformations ($O$ is an arbitrary operator)
$$
  \gamma_n (O) :=  \e^{- A_n}   O     \e^{ A_n}  =    \e^{- \mathrm{ad}_{A_n}}   O. 
$$
and 
  $$
  \alpha_n(O)  :=   \int_0^1 \d s \,   \e^{- s A_n}   O     \e^{ s A_n}  =    \int_0^1 \d s \,   \e^{-s \mathrm{ad}_{A_n}}   O. 
  $$
 The latter involves a dummy time $s$ that has nothing to do with the cycle time $t$; the transformation $\alpha_k$ is defined pointwisely for any $t $ in the cycle. 
The use of $\alpha_n$ is that
$$
\e^{- A_n} \partial_t  \e^{ A_n} =   \alpha_n(\partial_t A_n)
$$
as one easily checks by an explicit calculation. If $A_n(t)$ for different $t$ would commute among themselves, then we would simply find back the familiar expression
$
\e^{- A_n} \partial_t  \e^{ A_n} =   \partial_t A_n$.
With the help of the above notation, we get
\begin{align}
W_{n+1} & = \alpha_n(\adjoint_{A_n} H^0)-\adjoint_{A_n}H^0+ \gamma_n(D_n)-D_n +\iu(  \alpha_n(\partial_t A_n)- \partial_t A_n)  \nonumber \\[2mm]
&= -(\alpha_n(V_n)-V_n)+ \gamma_n(V_n+D_n)-(V_n+D_n)   \label{eq: expression wn}
\end{align}
where we used \eqref{eq: determine a} to get the second line. 
Finally, we see that
$$
D_{n+1}=D_n+ \caP_0(W_{n+1}),\qquad V_{n+1} =(1-\caP_0)(W_{n+1})
$$

\section{Proofs: analysis}

\subsection{Bound on $A_n$ from $V_n$}
 The most important ingredient is the bound on $A_n$ from properties of $V_n$. We recall that $V_n=\sum_{S \in \caP_\Lambda} V_{n,S}$, with
$V_{n,S}$  strongly supported in $S$. Then we have
\begin{lemma} \label{lem: bound on a}
Write $\norm O(\cdot)\norm :=\sup_t \norm O(t)\norm$. Then 
$$
\sup_t || A_{n,S} || \leq   C_q(1+ (1/x) (\gamma|S|)^{\tau+q})\,   T\sup_t || V_{n,S} ||
$$
where $A_{n,S}$ is defined by substituting  $V_{n,S}$ for $V_{n}$  in \eqref{eq: a in terms of v} and $C_q$ only depends on $q$.
\end{lemma}
\begin{proof}
Call $O=V_{n,S}$ and decompose 
$$
O=\sum_{\bfm} O_{\bfm} :=\sum_{\bfm} \caP_{\bfm}(O)
$$
There are norms on local operators that make $\bm{\caN}$ into a $q$-tuple of Hermitian operators, and hence, by spectral calculus, for any function $f:\bbZ^z \to \bbC$
\begin{equation}\label{eq: wasteful estimate}
 f(\bm{\caN})O= \sum_{\bfm} f({\bfm}) O_{\bfm} 
\end{equation}
Moreover, since $\caP_{\bfm}$ is contractive in operator norm, we have
$$
\norm f(\bm{\caN})O \norm  \leq  \sum_{{\bfm}: O_{\bfm}\neq 0} |f({\bfm})|   \norm O \norm.
$$
Obviously $(\caV_{0,T}-1)^{-1}=f(\bm{\caN})$ and, by our Diophantine assumption on $\bfJ$: 
$$ 
|f(\bfm)| \leq \max_{n\in\bbZ}|T\bfm\cdot\bar{\bfJ}- 2\pi n|^{-1} \leq \frac{|\bfm|^{\tau}}{x}
$$ 
We get then 
$$
\norm f(\bm{\caN})O \norm  \leq  C_q(1/x) (\gamma|S|)^{\tau+q}  \norm O\norm
$$
Apart from this, we simply use the estimate $\norm \caV_{t_0,t} O\norm \leq \norm O \norm$ in \eqref{eq: a in terms of v}. This yields the claim.
\end{proof}

Note that \eqref{eq: wasteful estimate} seems wasteful at first sight because one would want to replace the sum over $\bfm$ by a supremum over $\bfm$, which would be justfied, for example, if the decomposition $O=\sum_{\bfm}O_{\bfm}$ were orthogonal with respect to a scalar product associated to $\norm \cdot \norm$. It is not clear to us to what extent one could improve on this bound.

\subsection{Main lemma's}
 The following lemma is our prime tool. It was proven by cluster expansions in \cite{abanin2017rigorous}. Since the only difference here is that we consider strong supports rather than the normal support, we omit the proof which carries over line per line. 
\begin{lemma}\label{lem: main}
Let $Z, Q$ be potentials and assume that $ 3 \norm Q\norm_{\kappa}    \leq  \delta\kappa := {\kappa}- {\kappa'}$. Then 
$$
\norm \e^{Q}Z\e^{-Q}-Z   \norm_{\kappa'}  \leq      \frac{18}{ \delta\kappa\kappa'} \norm Z\norm_{\kappa}    \norm Q\norm_{\kappa}.   
$$
Since $ \norm Z\norm_{\kappa'} \leq  \norm Z\norm_{\kappa}$, we also get
$$
\norm \e^{Q}Z\e^{-Q}   \norm_{\kappa'}  \leq  \big(1+      \frac{18}{ \delta\kappa{\kappa'}}   \norm Q\norm_{\kappa}\big)     \norm Z\norm_{\kappa}.
$$
\end{lemma}

Furthermore, we also need the following estimate, of which we omit the obvious proof.

\begin{lemma}\label{lem: support multiplier}
Let $Z,\widetilde{Z}$ be potentials such thatn for some $p\geq 1$
$$
\norm\widetilde{Z}_S\norm \leq |S|^p\norm Z_S\norm,
$$
then, for $\tilde{\kappa}<\kappa$
$$
\norm \widetilde{Z}\norm_{\tilde{\kappa}} \leq    \frac{(p/\e)^p}{(\kappa-\tilde{\kappa})^p}\norm Z\norm_{{\kappa}}.
$$
\end{lemma}

\subsection{Iterating bounds} \label{sec: bounds on transformed potentials}
We define $\kappa(n)$ for $ n\geq 1$ by
$$
\kappa(n)=\frac{\kappa_0}{1+\log(n)}
$$ 
and 
$$
 \delta\kappa(n) \equiv \min\large( \kappa(n+\tfrac{1}{2})- \kappa(n+1),  \kappa(n)-  \kappa(n+\tfrac{1}{2}) \large) 
$$
so that $\frac{1}{\delta\kappa(n)} \leq 2n \log^2(n+1)$.
%
%
We abbreviate
  $\norm \cdot \norm_{\kappa(n)}$ by   $\norm \cdot \norm_{n}$ and     $\norm \cdot \norm_{\kappa(n+\tfrac{1}{2})}$ by   $\norm \cdot \norm_{n+\tfrac{1}{2}}$. 
We set
$$
w(n):= \norm W_n \norm_n, \qquad  d(n):=    \norm D_n \norm_n.
$$
Using the expression \eqref{eq: expression wn} for $W_{n+1}$ and Lemma \ref{lem: main}, we then get, provided that  $6 \norm A_n \norm_{n+\tfrac{1}{2}} \leq \delta \kappa(n)$
$$
\norm W_n \norm_{n+1} \leq   \frac{18}{ \delta\kappa(n)\kappa(n+1)} \norm A_n \norm_{n+\tfrac{1}{2}} \left(\norm D_n \norm_{n} +   \norm V_n \norm_{n} \right)
$$
Using Lemma \ref{lem: bound on a} and \ref{lem: support multiplier}, we have 
$$
\norm A_n \norm_{n+\tfrac{1}{2}} \leq     \frac{C(1+\tfrac{1}{x})T}{(\delta \kappa(n))^{\tau+q}}   \norm V_n \norm_{n}
$$
where $C$ depends only on $\kappa_0,\gamma,q$ (also in the equations below). Combining the two previous bounds and using  $\norm V_n \norm_{n} \leq \norm D_n \norm_{n}$ and  $\norm V_n \norm_{n} \leq \norm W_n \norm_{n}$, we get 
$$
w(n+1) \leq  C \log^3(n+1) n^{\tau+q+1}  [(1+\tfrac{1}{x})T]    w(n)d(n) 
$$
 Similarly, we have
$$
 d(n+1) \leq d(n)+ w(n+1)
$$
We can continue the iteration provided that $d(n)\leq Cd(1)$.
Finally, from the definition of the parameter $g$ we have that $w(1)=d(1)=g$, and so we obtain
$$
\frac{w(n+1)}{w(n)} \leq    C \log^3(n+1) n^{\tau+q+1}    [(1+\tfrac{1}{x})gT ] 
$$
Therefore, the iteration can be continued up to $n=n_*$ with $n_*$ the maximal number satisfying   $C\log^3(n_*+1) n_*^{\tau+q+1} \epsilon <1  $.  We then obtain the bound on $\norm \hat V\norm_\kappa= \norm V_{n_*}\norm_{n_*} $ in the theorem, choosing $n_*=\lfloor\epsilon^{-1/p} \rfloor$ with $p>\tau+q+1$. 
The bound showing the proximity of $\hat Y$ to identity follows just as in \cite{abanin2017rigorous}.

\section{Proofs for the static case}

Recall that here we have (no time-dependence)
$$
H={H^{0}}+V_1+ D_1,\qquad  H_n={H^{0}}+V_n+ D_n
$$
with  $V_n,D_n$ satisfying $\caP_0(V_n)=0$ and $\caP_0(D_n)=D_n$. 
We perform a unitary transformation $e^{A_n}$
and we write
\beq  \label{eq: transformed hamiltonian indep}
e^{A_n} \left(   {H^{0}}+D_n+V_n  \right) e^{-A_n} =   {H^{0}}+ D_n+ ([A_n,{H^{0}}] +V_n) + W_{n+1}
\eeq
with $W_{n+1}$ to be identified later. 
We determine $A_n$ by 
\begin{equation}\label{eq: determine a indep}
 [A_n, {H^{0}}]+V_n=0.
\end{equation}
and from here onwards the algebra is identical. Namely we get 
$$
W_{n+1} =-(\alpha_n(V_n)-V_n)+\gamma_n(V_n+D_n)-(V_n+D_n)
$$
Just as in the time-dependent case, the main point is to get a bound on $\norm A_n\norm_n$ from $\norm V_n \norm$. From \eqref{eq: determine a indep}, we see that a solution for $A_n$ is given by 
$$
A_{n,S} = -(\adjoint_{H^0})^{-1} V_{n,S}
$$ 
Proceeding similarly as in the time-dependent case, we obtain here
$$
\norm A_{n,S} \norm \leq  \norm V_{n,S} \norm    \sum_{\bfm \in \bbZ^q, |\bfm|\leq \gamma |S|}  \frac{1}{|\bfJ\cdot\bfm|}
\leq   C(1+\tfrac1x)  \frac{1}{|\bfJ|}(\gamma |S|)^{\tau+q-1}  \norm V_{n,S} \norm 
$$
where the last bound follows from the
the Diophantine condition $\frac{|\bfJ|}{|\bfJ\cdot\bfm|} \leq |\bfm|^\tau/x$ and we have written $(1+\tfrac1x)$ instead of $\tfrac1x$ simply to remind ourselves that it is impossible to choose $x$ large. 
If we now copy the steps done in the section \ref{sec: bounds on transformed potentials}, we obtain

$$
\frac{w(n+1)}{w(n)} \leq    C \log^3(n+1) n^{\tau+q}    [(1+\tfrac{1}{x}) \frac{g}{|\bfJ|} ] 
$$
and so we can now take $p> \tau+q$ instead of $p> \tau+q+1$.

\section{Appendix}

\subsection{Diophantine conditions}

We recall the Diophantine condition
$$
\inf_{n\in\bbZ}|T\bfm\cdot\bar{\bfJ}- 2\pi n| \geq \frac{x}{|\bfm|^{\tau}}
$$
To continue, let us fix  $\bm{b}= \tfrac{T}{2\pi}\bar{\bfJ}$ and $\mathbf{\hat b}= \frac{\mathbf{b}}{b} $ with $b=|\mathbf{b}|$. 
 We assume $\mathbf{\hat b}$ to be uniformly distributed on the unit hypersphere, and we calculate the probability that, for a given $\bfm$, 
$$
\mathbb{P}( \inf_{n\in\bbZ} | \mathbf{b}\cdot\bfm-n| \leq \frac{x}{|\bfm|^{\tau}} ). 
$$
We split $\mathbf{\hat b}=\frac{a}{|\bfm|}\bfm+{\mathbf{\hat b}}_{\perp}$ and we evaluate this probability as 
\begin{align*}
& C_q\sum_n \int d a (1-a^2))^{\tfrac{q-1}{2}} \chi( | (ba|\bfm|-n)| \leq \frac{x}{|\bfm|^{\tau}}    ) \\
&\leq C_q  \sum_n \int d a (1-a^2))^{\tfrac{q-1}{2}} \chi( | (a-\frac{n}{b|\bfm|}| \leq \frac{C(q,\tau) x}{b |\bfm|^{\tau+1}}    ) \\
&\leq    b|\bfm|  \frac{C(q,\tau) x}{b |\bfm|^{\tau+1}} =    \frac{C(q,\tau) x}{ |\bfm|^{\tau}}
\end{align*}
and summing this over all of $\bfm \in\bbZ^q$ yields the estimate, for some constant   $C'(q,\tau)$, 
$$
\mathbb{P}\left(\forall \bfm \in \bbZ^q:   \inf_{n\in\bbZ} | \mathbf{b}\cdot\bfm-n| \geq \frac{x}{|\bfm|^{\tau}} \right) \leq 1- C'(q,\tau)x,\qquad  \text{provided that $\tau>q $  }
$$
In particular, note that the magnitude $b= |T\mathbf{\bar J}|$ did not influence our bound. 


\end{document}